\documentclass[letterpaper]{article}
\usepackage{ijcai16}
\usepackage{times}
\usepackage{helvet}
\usepackage{courier}

\usepackage{tikz}
\usepackage{caption}
\usepackage{comment, color, amsmath, amsthm, amssymb,amsfonts,mathrsfs,graphicx,enumerate,fixltx2e}
\usepackage[subfigure]{graphfig}

\newtheorem{lemma}{Lemma}
\newtheorem{theorem}{Theorem}

\usepackage[linesnumbered,ruled,vlined]{algorithm2e}

\newtheorem{innermech}{Mechanism}
\newenvironment{mechanism}[1]
  {\renewcommand\theinnermech{#1}\innermech}
  {\endinnermech}

\DeclareMathOperator*{\argmax}{arg\,max}
\newcommand{\rulesep}{\unskip\ \hrule\ }

\usepackage{graphicx}
\newcommand*{\Scale}[2][4]{\scalebox{#1}{$#2$}}%

\begin{document}
%
\title{Mechanism Design with Exchangeable Allocations}
\author{
Qiang Zhang\\
\texttt{csqzhang@gmail.com}\\
}
\maketitle

\begin{abstract}
	We investigate mechanism design without payments when agents have different types of preferences. Contrary to most settings in the literature where agents have the same preference, e.g. in the facility location games all agents would like to stay close to (or  away from) the facility, we demonstrate the limitation of mechanism design without payments when agents have different preferences by introducing \textit{exchanging phases}. We consider two types of exchanging phases. The first model is called \textit{central exchanges} where the exchanges are performed by a central authority. The other model is called \textit{individual exchanges} where agents exchange their outcomes by themselves. By using facility location games as an example, we provide a truthful mechanism that optimizes social welfare in central exchanges. We also provide a universally truthful randomized  mechanism that achieves at least a half of the optimal social welfare in individual exchanges.
\end{abstract}

\section{Introduction}
Mechanism design aims to achieve desirable social outcomes when the partial but necessary information are possessed by selfish agents. To resolve this problem, \emph{truthful} mechanisms are designed to provide selfish agents incentives for declaring their true private information.  
Incentives are usually injected by the payments in those mechanisms such as auctions. Recently, \emph{mechanism design without payments} was proposed in~\cite{procaccia2009approximate}. They showed that, when the domain of agents' preferences is highly structured, it is possible to design truthful mechanisms without resorting to payments and achieve (approximately) good outcomes at the same time. Since then, the framework of mechanism design without payments has been successfully applied to different applications, such as voting, facility locations and machine scheduling. In this paper, we show  a limitation of the framework of mechanism design without payments in some settings, especially where agents have different preferences. As there is no payment, after social outcomes (or allocations for agents) being chosen by mechanisms, selfish agents could  \emph{exchange their allocations with others to their mutual benefits}. As mechanism designers, we should care about the final social outcomes, resulting from these \emph{exchanging phases}. In this paper, we call this problem \textit{mechanism design with exchangeable allocations}. Now let us consider an illustrative example on facility locations games.
\paragraph{Facility location games.}
In a facility location game, the location of a facility is to be determined based on the preferences reported by selfish agents. It is often to  assume that agents have certain preference structures.
For example, agents have single-peaked preferences if they prefer to stay close to the facility~\cite{procaccia2009approximate}, agents have obnoxious preferences if they prefer to be far away from the facility~\cite{cheng2011mechanisms}, or agents have double-peaked preferences if they prefer to stay at particular distance away from the facility~\cite{aris2015doublepeak}. In those studies, the private information of agents are their locations, which may be easily verified by the central authority. Now we assume that the locations of agents are the public information but mechanisms do not know the preference structures of agents. We consider two different types of preferences as follows. 
\begin{itemize}
	\item `Like'  preference (or single-peaked preference): agents would like to minimize their distances to the facility. We denote this preference by $L$.
	\item `Dislike' preference (or obnoxious preference): agents would like to maximize their distances away from the facility. We denote this preference by $H$.
\end{itemize}
Now let us illustrate the mechanism design with exchangeable allocations by the instance in Figure~\ref{fig:exp1}, where five agents are located in a line segment of a length of $8$ km. The locations of agents are public information. Agents are asked to report their preferences, like ($L$) or dislike ($H$). Figure~\ref{fig:exp1:a} shows the distances between agents and their true preferences. In this case, if we follow the approach of mechanism design without payments, then it is easy to verify that locating the facility that optimizes social welfare is a truthful mechanism. That is, there is no incentive for agents to misreport their preferences regardless the preferences reported by other agents. Thus, the facility will be located at  $5$km (the location of agent $3$) in this instance. As we mentioned before, after the facility being built, selfish agents could trade their houses. We consider that the classical top trading cycle (TTC) algorithm~\cite{shapley1974cores} is used here. Therefore, agent $4$ could trade her location with agent $1$, which results in that agent $4$ is $5$km away from the facility. Now, consider the same instance except that agent $4$ misreports her preference, shown as in Figure~\ref{fig:exp1:b}. To maximize the social welfare, the facility is built at $7$km (the location of agent $4$). Yet, once built, agent $4$ could trade her location with agent $1$ so that she is $7$km away from the facility. It provides her a better outcome than in Figure~\ref{fig:exp1:a} when she reported truthfully. Readers may notice that this situation would not happen if all agents have the same preference, $L$ or $H$. However, as mechanism designers, we are interested in dealing with the case when selfish agents have different preferences or even conflicted preferences. To summarize, using facility location games, we demonstrate a problem in the approach of mechanism design without payments. The insight is that mechanism designers or social planners should consider what will happen after the outcomes or allocations are chosen by mechanisms. Here we consider that agents are able to exchange their allocations for their mutual benefit. 

\begin{figure}[h]%
	\centering 
	\subfigure[Agent $a_4$ reports her preference ``L'' truthfully]{
		\label{fig:exp1:a}
		\begin{tikzpicture}[dot/.style={circle,inner sep=1pt,fill,label={#1}},
		extended line/.style={shorten >=-#1,shorten <=-#1},
		extended line/.default=1cm]
		\draw (0,1.3) -- (8,1.3);
		\node [draw=none] at (0.3,1) {$0$ km};
		\node [dot] at (0,1.3) {};
		
		\node [dot] at (1,1.3) {};
		\node [dot] at (2,1.3) {};
		\node [dot] at (3,1.3) {};
		\node [dot] at (4,1.3) {};
		\node [dot] at (5,1.3) {};
		\node [dot] at (6,1.3) {};
		\node [dot] at (7,1.3) {};
		
		\node [draw=none] at (7.7,1) {$8$ km};
		\node [dot] at (8,1.3) {};

		\node [dot=$L$] at (0,0) {};
		\node[draw=none] at (0,-0.3) {$a_1$};
		\node[draw=none] at (0.5,-0.3) {1};
		
		\node [dot=$H$] at (1,0) {};
		\node[draw=none] at (1,-0.3) {$a_2$};
		\node[draw=none] at (3,-0.3) {4};
		
		\node [dot=$L$] at (5,0) {};
		\node[draw=none] at (5,-0.3) {$a_3$};
		\node[draw=none] at (6,-0.3) {2};
		
		\node [dot=$H$] at (7,0) {};
		\node[draw=none] at (7,-0.3) {$a_4$};
		\node[draw=none] at (7.5,-0.3) {1};
		
		\node [dot=$L$] at (8,0) {};
		\node[draw=none] at (8,-0.3) {$a_5$};
		\draw (0,0) -- (8,0);
		\end{tikzpicture}}
	\subfigure[Agent $a_4$ misreports her preference as ``H'']{
		\label{fig:exp1:b}
		\begin{tikzpicture}[dot/.style={circle,inner sep=1pt,fill,label={#1},name={#1}},
		extended line/.style={shorten >=-#1,shorten <=-#1},
		extended line/.default=1cm]
		\node [dot=$L$] at (0,0) {};
		\node[draw=none] at (0,-0.3) {$a_1$};
		\node[draw=none] at (0.5,-0.3) {1};
		
		\node [dot=$H$] at (1,0) {};
		\node[draw=none] at (1,-0.3) {$a_2$};
		\node[draw=none] at (3,-0.3) {4};
		
		\node [dot=$L$] at (5,0) {};
		\node[draw=none] at (5,-0.3) {$a_3$};
		\node[draw=none] at (6,-0.3) {2};
		
		\node [dot=$L$] at (7,0) {};
		\node[draw=none] at (7,-0.3) {$a_4$};
		\node[draw=none] at (7.5,-0.3) {1};
		
		\node [dot=$L$] at (8,0) {};
		\node[draw=none] at (8,-0.3) {$a_5$};
		\draw (0,0) -- (8,0);
		\end{tikzpicture}}
	\caption{Example 1}
	\label{fig:exp1}
\end{figure}
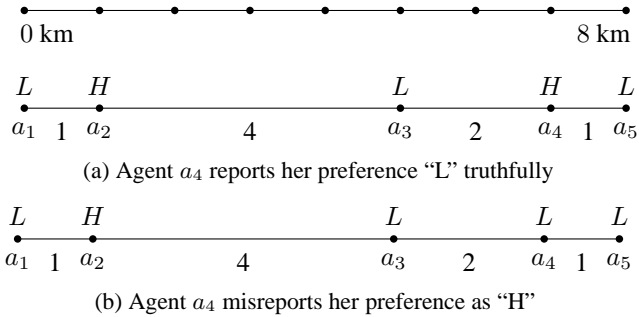	
\paragraph{Cooperative Games.} 
Cooperative games study how agents form a coalition and benefit together. The focus of mechanism design with exchangeable allocations is \textit{not} forming coalition. For example, in the facility location game in Figure~\ref{fig:exp1}, it is not necessary for agent $4$ to have a consensus with agent $1$ before her misreport. The assumption  we made is that agents are utility maximizers. Otherwise, agent $1$ may prefer to exchange her location with other agents rather than agent $4$.

\paragraph{Top Trading Cycle.} In this paper, we use the top trading cycle (TTC) algorithm~\cite{shapley1974cores} to simulate the exchange phase. There are different exchange protocols to trade individual good, e.g. deferred acceptance algorithm. The reasons that we choose the TTC algorithm over other choices are two-fold. First, the final outcome produced by the TTC algorithm is a core outcome, i.e., there is no other outcome in which all agents could be better off by exchanging their locations in a different way. It also implies that no exchange is possible afterwards. Second, as we are interested in truthful mechanisms, it is essential to have truthful exchange protocols. It is well-known that it is never advantageous to an agent to lie about its preference in the TTC algorithm. 


\subsection{Our approaches and results.}
We investigate mechanism design without payments when agents are able to exchange their allocations. We first formalize the problem as mechanism design with exchangeable allocations. Then, we study it in facility location games. The first natural approach is to take exchanging phase into account when we design truthful mechanisms. In other words, we are interested in mechanisms that could  exchange the locations of agents on their behalf. We call it central exchanges model. In this model, a mechanism is truthful if no agent could benefit by misreporting her preference and then exchanging her allocation with other agents. The first observation is that the design of such mechanism will highly depend on the exchange procedure. In this paper, we assume that the TTC algorithm is used in exchange phases. 
In addition, we consider a different benchmark. As agents may exchange their locations, better benchmarks would be the optimal social outcomes when one is able to relocate agents. In central exchanges model, we give a truthful mechanism that achieves the optimal social welfare.


Second, we understand that it is a strong assumption to allows mechanisms  to relocate agents. In Section~\ref{sec:individualExchange}, we investigated the case that mechanisms only determine the location of the facility and the exchanges are carried out by agents themselves. We call it individual exchanges model. We provide a simple universally truthful randomized mechanism, which guarantees that, after agents exchanging their locations, the social welfare it achieves is at least half of the social welfare achieved by any mechanism. Note that we do not require that agents exchanges their locations according to the TTC algorithm in individual exchanges model. We only assume that agents are utility maximizers and exchange for their mutual benefit. 


\subsection{Related Work}
Mechanism design without payments was coined  in~\cite{procaccia2009approximate} to study mechanism design problems where the goal is to optimize an objective under the restriction that no money transfer is possible. The problem considered in~\cite{procaccia2009approximate} is a facility location game where selfish agents with single-peaked preferences located on a line segment would like to minimize their distances to the facility.  They provided different deterministic and randomized truthful mechanisms for minimizing the total distance of agents or minimizing the maximum distance among agents. Later there is a significant amount of work
on facility location games with single-peaked preferences in literature both from social choice and computer science, e.g.~\cite{schummer2002strategy,lu2009tighter,lu2010asymptotically,fotakis2010winner,fotakis2012power,serafino2015truthful}.

Besides single-peaked preferences, facility locations games are also investigated under other preferences. For example, \cite{aris2015doublepeak} studied the case where agents have double-peaked preferences, in which agents would like to stay neither too close or too far away from the facility. \cite{zou2015facility} studied facility location games with dual preference where agents are either willing to stay close to or willing to stay away from the facility.

To the best of our knowledge, all existing work on facility location games deal with the case that the locations of agents are private information, i.e., the true locations of agents are only known to themselves. The primary reason is that most papers assume the preferences of agents and their preferences is unique. Arguably, the locations of agents are relatively easy to verify by central authority. Thus, in this paper we assume that there are two types of preference. However, the exact types of preferences are agents' private information. It is the first paper to study mechanism design on facility location games with different preferences when preferences are unknown to mechanism.

\section{Mechanism Design with Exchangeable Allocations}
\label{sec:MDTrading}
In this section, we present the general description on mechanism design with exchangeable allocations, which consists of two phases. The first phase is similar to mechanism design without payments in~\cite{procaccia2009approximate}. There is a central authority who aims to optimize an public known objective $\mathcal{F}$ by choosing an outcome from a set of all possible outcomes denoted by $O$. To achieve good outcomes, the authority implements a mechanism, say $\mathcal{M}$, which first asks a set of selfish agents, denoted by $A$, to  declare their private information and then chooses an outcome. In addition,  each agent $i \in A$ has an utility function $u_i(\cdot)$ for all possible outcomes. Let $c_i$ denote the private information of agent $i$ and $d_i$ be the private information declared by the agent $i$. Given $\mathbf{d}=<d_1,\ldots,d_n>$, Mechanism $\mathcal{M}$ chooses an outcome denoted by $\mathcal{M}(\mathbf{d})$. It is possible that agents would receive different outcomes, for example, agents would receive different resource in resource allocation problems. To allow the generality, we use $\mathcal{M}_i(\mathbf{d})$ to denote the outcome for agent $i$. Hence, the utility of agent $i$ is $u_i(\mathcal{M}_i(\mathbf{d}))$.\footnote{We assume there is no externalities.} 

The second phase involves an (possibly pre-defined) exchange procedure $\mathcal{E}$. Agents improve their utilities by exchanging their outcomes via procedure $\mathcal{E}$. For example, agent $i$ may exchange her outcome with agent $j$ if both agent $i$ and $j$ improve their utility. That is,
\begin{align*}
u_i(\mathcal{M}_j(\mathbf{d})) > u_i(\mathcal{M}_i(\mathbf{d}))  \mbox{ and } 
u_j(\mathcal{M}_i(\mathbf{d})) > u_j(\mathcal{M}_j(\mathbf{d})) 
\end{align*}
Note that the exchange procedure is not fixed and can be specified due to different purposes. For example, as we mentioned, in this paper, we are mainly interested in the top trading circle as the exchange procedure. Let $\mathcal{E}_i(\mathcal{M}(\mathbf{d}))$ be the outcome of agent $i$ after the exchanges. Hence, in the end agent $i$ has an utility $u_i(\mathcal{E}_i(\mathcal{M}(\mathbf{d})))$.

The goal of mechanism design with exchangeable allocations is to design mechanism that give agents incentive to declare their true private information regardless the information declared by others. Specifically , a mechanism $\mathcal{M}$ is truthful if it holds that
\[
u_i(\mathcal{E}_i(\mathcal{M}(\mathbf{d}_{-i},c_i)) \geq u_i(\mathcal{E}_i(\mathcal{M}(\mathbf{d}_{-i},d_i)) 
\]
for any $i, \mathbf{d}_{-i}$ and $c_i$ where $\mathbf{d}_{-i}=<d_1,\ldots, d_{i-1},d_{i+1},\ldots, d_n>$.

In the rest of paper, we use facility location games to demonstrate two different approaches to design truthful mechanism. 


\section{Central Exchanges}
In this section, we will focus on the central exchanges approach in facility location games, where mechanisms could assign new locations to agents after the location of facility has been decided. As we are interested in mechanisms that maximize social welfare, it implies that agents will have no incentives (or will not want) to exchange location themselves given their new locations.

\subsection{Model}
In a facility location game, a facility is to built in a line segment $[0,d]$, where $n$ agents are located. Agent $i$ is at point $x_i$ that is known to the public, and has a private preference (aka, type), Like ($L$) or Dislike $(H)$, denoted by $t_i$.  We will call the collection $\mathbf{x}=<x_1, \ldots, x_n>$ a \emph{location profile}, and the collection $\mathbf{t} = <t_1, \ldots, t_n>$ a \emph{preference profile}.
Furthermore, we say $\hat{\mathbf{x}}$ is \textit{permuted}  to $x$ if there exists a permutation $\pi$ such that for any $i \in [1,n]$ it holds $\hat{x}_i = x_{\pi(i)}$.

A deterministic mechanism is a function $\mathcal{M}: \mathbb R^{2n} \mapsto \mathbb R^{n+1}$ mapping a location profile $\mathbf{x}$ and  a preference profile $\mathbf{t}$ to a permuted location profile and a point in $\mathbb R$ which is the location of the facility. For the purpose of easy presentation, we sometimes denote the location of the facility by $\mathcal{M}^f(\mathbf{x}, \mathbf{t})$ and the permuted location profile by  $\mathcal{M}(\mathbf{x}, \mathbf{t})$. We also denote the location assigned to agent $i$ by $\mathcal{M}_i(\mathbf{x}, \mathbf{t})$.
Intuitively speaking, a mechanism builds the facility at point $\mathcal{M}^f(\mathbf{x}, \mathbf{t})$ and relocates agent $i$ to location $\mathcal{M}_i(\mathbf{x}, \mathbf{t})$. 
Given the location $\hat{x}_i$ assigned to agent $i$ and the location $y \in \mathbb R$ of the facility,  
the utility of agent $i$ is 
\begin{align*}
\mathrm{u}(y, \hat{x}_i, t_i) =
\begin{cases}
d - |\hat{x}_i-y| & \text{if}\ t_i = L \\
|\hat{x}_i - y| & \text{if}\ t_i = H
\end{cases}
\end{align*}

Given a permuted location profile $\hat{\mathbf{x}}$ and a location of the facility $y \in \mathbb R$, the social welfare is \
\[
\mathrm{SW}(y,\hat{\mathbf{x}},\mathbf{t}) = \sum_{i \in [1,n]}\mathrm{u}(y, \hat{x}_i, t_i)
\]

A deterministic mechanism is \emph{truthful} if no agent would benefit by misreporting her preference regardless the locations and preferences declared by other agents. That is, for every possible location profile $\mathbf{x}$, every possible preference profile $\mathbf{t}$, every agent $i \in [1,n]$ and every $t'_i \in \{L,H\}$, it holds that $u(\mathcal{M}^f(\mathbf{x}, \mathbf{t}), \mathcal{M}_i(\mathbf{x}, \mathbf{t}), t_i) \geq u(\mathcal{M}^f(\mathbf{x}, <\mathbf{t}_{-i},t'_i>), \mathcal{M}_i(\mathbf{x}, <\mathbf{t}_{-i},t'_i>), t'_i)$, where $\mathbf{t}_{-i} = <t_1, \ldots, t_{i-1}, t_{i+1},\ldots, t_n>$.

We are interested in truthful mechanisms that perform well with respect to the optimal social welfare. The performance ratio $r$ of a mechanism is given by comparing the social welfare it achieves with the optimal social welfare on any location profile $\mathbf{x}$ and any preference profile $\mathbf{t}$ as follows. 
\[
r = \max_{y,\mathbf{x}, \hat{\mathbf{x}}, \mathbf{t}} \frac{\mathrm{SW}(y, \hat{\mathbf{x}},\mathbf{t})}{\mathrm{SW}(\mathcal{M}^f(\mathbf{x}, \mathbf{t}),\mathcal{M}(\mathbf{x}, \mathbf{t}),\mathbf{t}) }
\]
where $\hat{\mathbf{x}}$ is permuted to $\mathbf{x}$. 

\subsection{Optimal Mechanism}
We first present a customized location permutation (see Algorithm~\ref{LocationPermuation}), which returns a location profile that is permuted to $\hat{\mathbf{x}}$ given location profiles $\mathbf{x}$ and $\hat{\mathbf{x}}$, 
a preference profile $\mathbf{t}$ and a facility location $y$. 

\begin{algorithm}
	Let $\mathbf{L}$ and $\mathbf{H}$ denote the set of agents whose preferences are $L$ and $H$, respectively\;
	Let $\hat{\mathbf{x}}_{\mathbf{L}} $ and $\hat{\mathbf{x}}_{\mathbf{H}} $ denote the locations of agents in $\mathbf{L}$ and $\mathbf{H}$ in $\hat{\mathbf{x}}$, respectively\; 
	Sort agents in $L$ by their distances to $y$ in $\mathbf{x}$ in an increasing order. Similarly, sort agents in $H$ by their distances to $y$ in $\mathbf{x}$ in a decreasing order\;
	Sort locations in $\hat{\mathbf{x}}_{\mathbf{L}} $ by the distances to $y$ in an increasing order. Similarly, sort locations in $\hat{\mathbf{x}}_{\mathbf{H}} $ by the distances to $y$ in a decreasing order\;
	Assign agents in $L$ sequentially to locations in $\hat{\mathbf{x}}_{\mathbf{L}} $ based on their orders, and similarly assign agents in $H$ sequentially to locations in $\hat{\mathbf{x}}_{\mathbf{H}} $\;
	Return the assignment.
	\caption{LocationPermuation($\mathbf{x}, \hat{\mathbf{x}}, \mathbf{t},y$)}
	\label{LocationPermuation}	
\end{algorithm}
Essentially, Algorithm~\ref{LocationPermuation} assigns the locations $\hat{\mathbf{x}}_{\mathbf{L}}$(or $\hat{\mathbf{x}}_{\mathbf{H}}$) to agents in $\mathbf{L}$(or $\mathbf{H}$) according to their distances to the facility in location profile $\mathbf{x}$. Now, we present the following mechanism:
\begin{mechanism}{$\mathcal{OPT}$}\label{optMechanism}
	Given a location profile $\mathbf{x}$ and a preference profile $\mathbf{t}$, choose a permuted locations $\mathbf{z}$ and $y$ such that $\mathrm{SW}(y,\mathbf{z},\mathbf{t})$ is maximized, tie-breaking arbitrarily.  The mechanism returns $y$ as the location of the facility $\mathcal{OPT}^f(\mathbf{x}, \mathbf{t})$, and call $\mathrm{LocationPermuation(\mathbf{x}, \mathbf{z}, \mathbf{t}, y)}$ to compute the locations for agents $\mathcal{OPT}(\mathbf{x}, \mathbf{t})$. 
\end{mechanism}

Note that $z$ and $y$ may not be unique. We assume that Mechanism~\ref{optMechanism}  uses an arbitrary tie-breaking rule. The tie-breaking rule has no effect on the performance of the mechanism as it always maximizes the social welfare. As we will show that, if agents could benefit by lying their preferences in Mechanism~\ref{optMechanism}, then the social welfare is not maximized. Therefore, any arbitrary tie-breaking rule works for the truthfulness. 

\paragraph{Examples of Algorithm~\ref{LocationPermuation}.} Now we would like to illustrate Mechanism~\ref{optMechanism} using the instance in Figure~\ref{fig:exp1}. Readers could verify that $\mathrm{SW}(y,\mathbf{z},\mathbf{t})$ is maximized when $y = 8$ and $\mathbf{z}=\{z_1 = 5, z_2 = 1, z_3=8, z_4 = 0, z_5 = 7 \}$.  Mechanism~\ref{optMechanism} locates the facility at $8$. Now let us see the locations assigned to agents. As agent $5$ is close to the location of facility among agents with preference $L$, Mechanism~\ref{optMechanism} assigns agent $5$ to location $8$. As agent $3$ is the second closest to the location of facility after agent $5$, Mechanism~\ref{optMechanism} assigns agent $3$ to location $7$. Next, Mechanism~\ref{optMechanism} assigns agent $1$ to location $5$. Similarly,  for agents with preference $H$, Mechanism~\ref{optMechanism} assigns agent $4$ to location $1$ and agent $2$ to location $0$.

\paragraph{Computing $\mathbf{z}$ in~\ref{optMechanism}.} Mechanism~\ref{optMechanism} first chooses the location of the facility $y$ and a location profile $\mathbf{z}$ that is permuted to $\mathbf{x}$ such that the social welfare is maximized. 
In the following we show that it is computationally feasible to find a $y$  and $\mathbf{z}$ given $\mathbf{x}$ and $\mathbf{t}$. 

\begin{lemma}\label{lem:consecutive}
	Given $\mathbf{x}$, $\mathbf{t}$, $y$ and $\mathbf{z}$ that is permuted to $\mathbf{x}$, if $\mathrm{SW}(y, \mathbf{z}, \mathbf{t})$ is maximized, then for any $i \in \mathbf{H}$, $z_i \leq \min \mathbf{z}_{\mathbf{L}}$ or $z_i \geq \max \mathbf{z}_{\mathbf{L}}$.
\end{lemma}
Intuitively, Lemma~\ref{lem:consecutive} implies that agents whose preferences are $L$ are assigned to consecutive locations  where social welfare is maximized. This characterization strictly reduces the space of the optimal outcomes. Next we present another characterization on the locations of the facility in the optimal outcomes.
\begin{lemma}\label{lem:optimalLocations}
	Given $\mathbf{x}$, $\mathbf{t}$ and $\mathbf{z}$ that is permuted to $\mathbf{x}$, there exists a $y$ such that $y \in \mathbf{x}\cup \{0,d\}$ and $\mathrm{SW}(y,\mathbf{z}, \mathbf{t})$ is maximized.
\end{lemma}

By Lemma~\ref{lem:consecutive} and~\ref{lem:optimalLocations}, we limit the search space for the optimal outcomes to $O(n^2)$. The optimal outcomes can be found by enumerating all possible outcomes and returning the outcome with the maximum social welfare. If ties exist, Mechanism~\ref{optMechanism} breaks ties arbitrarily.

\begin{lemma}
	\label{lem:truthfulness}
	Mechanism~\ref{optMechanism} is truthful.
\end{lemma}
\begin{proof}
	We show that, given any $\mathbf{x}, \mathbf{t}$, any agent $i \in H$ cannot increase her utility by declaring $t_i = L$. The proof for the opposite case that any agent $i \in L$ cannot increase her utility by declaring $t_i = H$ is symmetric and omitted due to the space limit. 
	
	
	\begin{figure}[h]%
		\centering 
		\rulesep
		\subfigure[The outcome when agent $a_i$ reports her preference ``H'' truthfully in $\mathcal{M}_1$]{
			\begin{tikzpicture}[dot/.style={circle,inner sep=1pt,fill,label={#1},name={#1}},
			extended line/.style={shorten >=-#1,shorten <=-#1},
			extended line/.default=1cm]
			\label{fig:truthfulProof:a}
			\node [dot=$H$] at (0,0) {};
			
			\node[draw=none] at (0.8,0) {$\ldots$};
			
			\node [dot=$H$] at (1.6,0) {};
			\node[draw=none] at (1.6,-0.3) {$z_i$};
			
			\node[draw=none] at (2.4,0) {$\ldots$};
			
			\node [dot=$H$] at (3.2,0) {};
			
			\node [dot=$L$] at (3.6,0) {};
			
			\node[draw=none] at (4.4,0) {$\ldots$};
			
			\node [dot=$L$] at (5.2,0) {};
			\node[draw=none] at (5.2,-0.3) {$y=\mathcal{M}_1^f(\mathbf{x}, \mathbf{t})$};
			
			\node[draw=none] at (6.0,0) {$\ldots$};
			
			\node [dot=$L$] at (6.5,0) {};
			
			\node [dot=$H$] at (7.1,0) {};
			
			\node[draw=none] at (7.5,0) {$\ldots$};
			
			\end{tikzpicture}
		}
		\rulesep	
		\subfigure[A possible outcome when agent $a_i$ reports her preference as  $L$]{
			\begin{tikzpicture}[dot/.style={circle,inner sep=1pt,fill,label={#1},name={#1}},
			extended line/.style={shorten >=-#1,shorten <=-#1},
			extended line/.default=1cm]
			\label{fig:truthfulProof:b}
			\node [dot=$H$] at (0,0) {};
			
			\node[draw=none] at (0.8,0) {$\ldots$};
			
			\node [dot=$H$] at (1.6,0) {};
			\node[draw=none] at (1.6,-0.3) {$\bar{z}_i$};

			\node[draw=none] at (2.4,0) {$\ldots$};
			
			\node [dot=$H$] at (3.2,0) {};

			\node [dot=$L$] at (3.6,0) {};
			
			\node[draw=none] at (4.4,0) {$\ldots$};
			
			\node [dot=$L$] at (5.2,0) {};
			\node[draw=none] at (5.2,-0.3) 
			{$\bar{y}$};
			
			\node[draw=none] at (6.0,0) {$\ldots$};
			
			\node [dot=$L$] at (6.5,0) {};
			
			\node [dot=$H$] at (7.1,0) {};
			
			\node[draw=none] at (7.5,0) {$\ldots$};
			
			\end{tikzpicture}
		}	
		\rulesep
		\subfigure[The outcome the maximize social welfare when agent $a_i$ misreports her preference as $L$ in $\mathcal{M}_1$]{
			\label{fig:truthfulProof:c}
			\begin{tikzpicture}[dot/.style={circle,inner sep=1pt,fill,label={#1},name={#1}},
			extended line/.style={shorten >=-#1,shorten <=-#1},
			extended line/.default=1cm]
			
			\node[draw=none] at (0.0,0) {$\ldots$};
			\node [dot=$H$] at (0.6,0) {};
			
			\node [dot=$L$] at (1,0) {};
			
			\node[draw=none] at (1.6,0) {$\ldots$};
			
			\node [dot=$L$] at (2.2,0) {};
			\node[draw=none] at (2.2,-0.3) {$z'_i$};
			
			\node[draw=none] at (2.8,0) {$\ldots$};
			
			\node [dot=$L$] at (3.8,0) {};
			\node[draw=none] at (3.8,-0.3) {$y'=\mathcal{M}_1^f(\mathbf{x}, \mathbf{t})$};
			
			\node[draw=none] at (4.4,0) {$\ldots$};
			\node[draw=none] at (5.2,0) {$\ldots$};
			\node[draw=none] at (6,0) {$\ldots$};
			
			\node [dot=$L$] at (6.4,0) {};
			
			\node[draw=none] at (6.8,0) {$\ldots$};
			
			\node [dot=$H$] at (7.4,0) {};
			\node[draw=none] at (8,0) {$\ldots$};
			\end{tikzpicture}
		}
		\rulesep
		\subfigure[The outcome after agent $a_i$ misreporting and exchanging his location with other agent]{
			\label{fig:truthfulProof:d}
			\begin{tikzpicture}[dot/.style={circle,inner sep=1pt,fill,label={#1},name={#1}},
			extended line/.style={shorten >=-#1,shorten <=-#1},
			extended line/.default=1cm]
			
			\node[draw=none] at (0.0,0) {$\ldots$};
			\node [dot=$H$] at (0.6,0) {};
			
			\node [dot=$L$] at (1,0) {};
			
			\node[draw=none] at (1.6,0) {$\ldots$};
			
			\node [dot=$L$] at (2.2,0) {};
			
			\node[draw=none] at (2.8,0) {$\ldots$};
			
			\node [dot=$L$] at (3.8,0) {};
			\node[draw=none] at (3.8,-0.3) {$\bar{y}'$}; 
			
			\node[draw=none] at (4.4,0) {$\ldots$};
			\node[draw=none] at (5.2,0) {$\ldots$};
			\node[draw=none] at (6,0) {$\ldots$};
			
			\node [dot=$L$] at (6.4,0) {};
			\node[draw=none] at (6.4,-0.3) {$\bar{z}'_i$};

			\node[draw=none] at (6.8,0) {$\ldots$};
			
			\node [dot=$H$] at (7.4,0) {};
			\node[draw=none] at (8,0) {$\ldots$};
			\end{tikzpicture}
		}
		\rulesep
		%
		%
		%
		%
		%
		%
		%
		%
		%
		%
		%
		%
		\caption{Instances in the proof of Lemma~\ref{lem:truthfulness}}
		\label{fig:truthfulInstances}
	\end{figure}
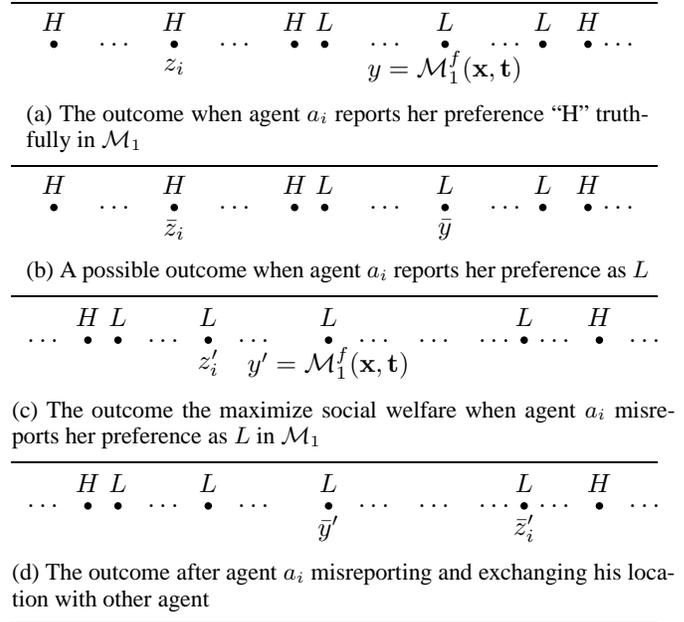
	
	Before presenting the proof, we introduce the following four instances that will be used in this proof. The instances are shown in Figure~\ref{fig:truthfulInstances},
	\begin{itemize}
		\item Instance $a$: $\mathbf{x}$ and $\mathbf{t}$. Outcome: $<y, \mathbf{z}>$. In this instance, agents report preference truthfully. $\mathcal{M}_1$ places the facility at $y$ and reassigns each agent $j \in A$ to $z_j$. The social welfare is maximized. 
		\item Instance $b$: $\bar{\mathbf{x}}$ and $\bar{\mathbf{t}}$ where $\bar{\mathbf{x}} = \mathbf{x}$ and $\bar{\mathbf{t}} = \mathbf{t}$ except $\bar{t}_i = L$. Outcome: $<\bar{y}, \bar{\mathbf{z}>}$ where $\bar{y} = y$ and $\bar{\mathbf{z}}=\mathbf{z}$. In this instance, agent $i$ misreports her preference.  The facility is placed at $\bar{y}$ and agent $j \in A$ is assigned to $\bar{z}_j$.
		\item Instance $c$: $\mathbf{x}'$ and $\mathbf{t}'$ where $\mathbf{x}' = \bar{x}$ and $\mathbf{t}' = \bar{\mathbf{t}}$. Outcome: $<y', \mathbf{z}'>$. In this instance,  Mechanism~\ref{optMechanism} places the facility at $y'$ and assigns each agent $j \in A$ to $z'_j$. The social welfare is maximized.
		\item Instance $d$: $\bar{\mathbf{x}}'$ and $\bar{\mathbf{t}}'$ where $\bar{\mathbf{x}}' = x$ and $\bar{\mathbf{t}}' = \mathbf{t}$. Outcome: $<\bar{y}', \bar{\mathbf{z}}'>$. This instance is essentially the exchange phase, where agent $i$ exchanges his location in $\mathbf{z}'$ with 
		agent $j$ who is furthest to $y'$ with preference $L$ in $\mathbf{z}'$.
	\end{itemize}

	From these instances , we get the following:
	\begin{align}
	&\sum_{j \in H \setminus \{ i\}} u(\bar{y}, \bar{z}_j, \bar{t}_i) + \sum_{j \in L \setminus \{ i\}} u(\bar{y}, \bar{z}_j, \bar{t}_i) + u(\bar{y}, \bar{z}_i,  L) \nonumber \\
	\leq & \sum_{j \in H \setminus \{ i\}} u(y', z'_j, t'_i) + \sum_{j \in L \setminus \{ i\}} u(y', z'_j, t'_i) + u(y', z'_i, L)   \nonumber \\
	= & \sum_{j \in H \setminus \{ i\}} u(\bar{y}', \bar{z}'_j, \bar{t}'_i) + \sum_{j \in L \setminus \{ i\}} u(\bar{y}', \bar{z}'_j, \bar{t}'_i) + u(\bar{y}', \bar{z}'_i, L)  \label{eq:SumSocialWelfare}
	\end{align}
	
	The first inequality comes from that $y'$ and $\mathbf{z}'$ maximizes the social welfare when there are $|L|+1$ agents with preference $L$. \footnote{In instances $b$ and $c$ the preference of agent $i$ is $L$} 
	Now let us assume that agent $i$ could benefit by misreporting her preference as $L$ and then exchanging her location with other agent. As
	$ u(\bar{y},\bar{z_i},H)  = u(y,z_i,H)$ by definition, it implies that $u(\bar{y}',\bar{z}'_i,H) > u(\bar{y},\bar{z_i},H)$. Since the relation of two types of preferences are opposite, we know that that $u(\bar{y}',\bar{z}'_i,L) < u(\bar{y},\bar{z}_i,L)$. Substitute it into Equation~(\ref{eq:SumSocialWelfare}), we get 
	\begin{align}
	&\sum_{j \in H \setminus \{ i\}} u(\bar{y}, \bar{z}_j, \bar{t}_i) + \sum_{j \in L \setminus \{ i\}} u(\bar{y}, \bar{z}_j, \bar{t}_i)  \nonumber \\
	< & \sum_{j \in H \setminus \{ i\}} u(\bar{y}', \bar{z}'_j, \bar{t}'_i) + \sum_{j \in L \setminus\{ i\}} u(\bar{y}', \bar{z}'_j, \bar{t}'_i) 
	\end{align}
	Hence, we have
	\begin{align*}
	&\sum_{j \in H \setminus \{ i\}} u(y, z_j, t_i) + \sum_{j \in L \setminus \{ i\}} u(y, z_j, t_i)  + u(y,z_i,H) \\
	=&\sum_{j \in H \setminus \{ i\}} u(\bar{y}, \bar{z}_j, \bar{t}_i) + \sum_{j \in L \setminus \{ i\}} u(\bar{y}, \bar{z}_j, \bar{t}_i)  + u(\bar{y},\bar{z_i},H) \nonumber \\
	< & \sum_{j \in H \setminus \{ i\}} u(\bar{y}', \bar{z}'_j, \bar{t}'_i) + \sum_{j \in L \setminus \{ i\}} u(\bar{y}', \bar{z}'_j, \bar{t}'_i) + u(\bar{y}',\bar{z}'_i,H)
	\end{align*}
	It contradicts the fact that  Mechanism \ref{optMechanism} maximizes the social welfare. 
\end{proof}

\section{Individual Exchanges}
\label{sec:individualExchange}
In this section, we explore the truthful mechanisms when social planners are less powerful. In particular, exchanges between agents are coordinated by agents themselves rather than mechanisms. We assume that the TTC algorithm simulates the exchanges between agents. The problem is formalized as follows. 

\subsection{Model}
The goal is to locate a facility in a line segment $[0,d]$ given a  location profile $\mathbf{x}$ and a preference profile $\mathbf{t}$, which are defined as previous section. 

A deterministic mechanism is a function $\mathcal{M}: \mathbb{R}^{2n} \rightarrow \mathbf{R}$ which is the location of the facility. A randomized mechanism is a function $\mathcal{M}: \mathbb{R}^{2n} \rightarrow \Delta(\mathbf{R})$ where $\Delta(\mathbf{R})$ is the set of distribution over $\mathbf{R}$. Given any location profile $\mathbf{x}$ and preference profile $\mathbf{t}$, let $\mathcal{M}(\mathbf{x}, \mathbf{t})$ denote the location of facility in mechanism $\mathcal{M}$.  $\mathcal{M}(\mathbf{x}, \mathbf{t})$ is a random variable in a randomized mechanism. 

Given the location of the facility $y$, the location profile $\mathbf{x}$ and the preference profile $\mathbf{t}$, let $\mathcal{TTC}(\mathbf{x},\mathbf{t},y)$ denote the outcome of the TTC algorithm. In particular, let $\mathcal{TTC}_i(\mathbf{x},\mathbf{t},y)$ be the location of agent $i$ from the TTC algorithm. Hence, the utility of agent $i$ is 
\begin{align*}
\Scale[0.92]{
	\mathrm{u}(y, \mathcal{TTC}_i(\mathbf{x},\mathbf{t},y), t_i) =
	\begin{cases}
	d - |\mathcal{TTC}_i(\mathbf{x},\mathbf{t},y)-y| & \text{if}\ t_i = L \\
	|\mathcal{TTC}_i(\mathbf{x},\mathbf{t},y) - y| & \text{if}\ t_i = H
	\end{cases}
}
\end{align*}
Similarly, the social welfare obtained by TTC is 
\[
\mathrm{SW}(y,\mathcal{TTC}(\mathbf{x},\mathbf{t},y),\mathbf{t}) = \sum_{i \in [1,n]}\mathrm{u}(y, \mathcal{TTC}_i(\mathbf{x},\mathbf{t},y), t_i)
\]

A deterministic mechanism $\mathcal{M}$ is \textit{truthful} if no agent would benefit by misreporting her preference in the mechanism and then obtaining a better outcome from the TTC algorithm regardless of the locations and preferences of other agents. That is, for any location profile $\mathbf{x}$, any preference profile $\mathbf{t}$, any $i \in [1,n]$ and any $t'_i \in \{L, H\}$, it holds $u(\mathcal{M}(\mathbf{x},\mathbf{t}), \mathcal{TTC}_i(\mathbf{x}, \mathbf{t}, \mathcal{M}(\mathbf{x},\mathbf{t})),t_i) \geq        
u(\mathcal{M}(\mathbf{x},<\mathbf{t}_{-i},t'_i>), \mathcal{TTC}_i(\mathbf{x}, <\mathbf{t}_{-i},t_i>, \mathcal{M}(\mathbf{x},<\mathbf{t}_{-i},t'_i>)),t_i)$. Note that the preference $t'_i$ misreported by agent $i$ only affects the location of the facility chosen by mechanism $\mathcal{M}$ only. Since the exchanges happened in the TTC algorithm are carried out by agents themselves, the true preference $t_i$ is used there. A universally truthful mechanism is a probability distribution over deterministic truthful mechanisms

As the previous section, we measure the performance of the truthful mechanism with its approximation ratio to the optimal social welfare.  Given mechanism $\mathcal{M}$, its approximation ratio is defined as follows.  
\[
r^{\mathcal{M}} = \max_{y,\mathbf{x}, \hat{\mathbf{x}}, \mathbf{t}} \frac{\mathrm{SW}(y, \hat{\mathbf{x}},\mathbf{t})}{ \mathrm{SW}(\mathcal{M}(\mathbf{x}, \mathbf{t}), \mathcal{TTC}(\mathbf{x}, \mathcal{M}(\mathbf{x}, \mathbf{t}),\mathbf{t}), \mathbf{t})  }
\]
where $\hat{\mathbf{x}}$ is permuted to $\mathbf{x}$.
\subsection{Optimal location}
From the previous section, we know that, given the location of the facility, agents with the same preference ($L$ or $H$) will be consecutive after their exchanges in the TTC algorithm. Hence, it is easy to verify that, assuming agents report their preference truthfully, one could locate the facility at the optimal location $y$ so that the optimal social welfare will be achieved after the exchanges between agents. Hence, the following question arises naturally: \textit{Is it a truthful mechanism by locating the facility at the optimal location conditioned on that agents would exchange their locations by TTC algorithm?} Note that in the example provided in Figure~\ref{fig:exp1} the facility is located at the optimal location assuming that agents \textit{would not} exchange their locations. 

\begin{theorem}
	\label{thm:optimalLocation}
	Given $\mathbf{x}$ and $\mathbf{t}$, $\mathcal{M}(\mathbf{x},\mathbf{t}) = \argmax_{y \in [0,d]} \mathrm{SW}(y, \hat{\mathbf{x}},\mathbf{t})$ is not truthful given that agents would exchange their locations by the TTC algorithm. 
\end{theorem}
\begin{proof}
	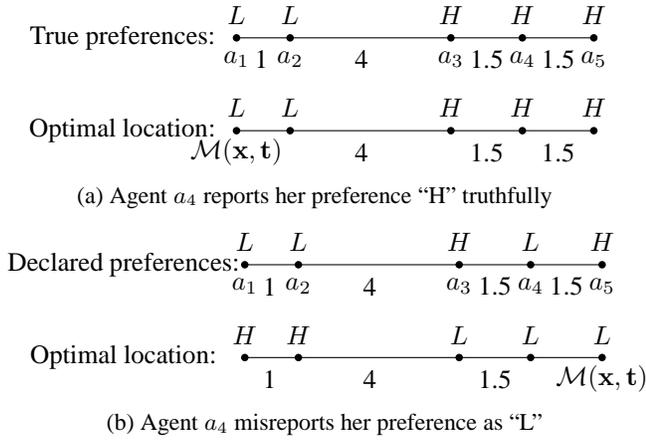
\begin{figure}[h]%
		\centering 
		\subfigure[Agent $a_4$ reports her preference ``H'' truthfully]{
			\label{fig:exp2:a}
			\begin{tikzpicture}[scale=0.95, dot/.style={circle,inner sep=1pt,fill,label={#1},name={#1}},
			extended line/.style={shorten >=-#1,shorten <=-#1},
			extended line/.default=1cm]
			
			\node[draw=none] at (-1.6,-0) {True preferences:};
			
			\node [dot=$L$] at (0,0) {};
			\node[draw=none] at (0,-0.3) {$a_1$};
			\node[draw=none] at (0.35,-0.3) {1};
			
			\node [dot=$L$] at (0.75,0) {};
			\node[draw=none] at (0.75,-0.3) {$a_2$};
			\node[draw=none] at (1.75,-0.3) {4};
			
			\node [dot=$H$] at (3,0) {};
			\node[draw=none] at (3,-0.3) {$a_3$};
			
			\node[draw=none] at (3.5,-0.3) {1.5};
			
			\node [dot=$H$] at (4,0) {};
			\node[draw=none] at (4,-0.3) {$a_4$};
			\node[draw=none] at (4.5,-0.3) {1.5};
			
			\node [dot=$H$] at (5,0) {};
			\node[draw=none] at (5,-0.3) {$a_5$};
			\draw (0,0) -- (5,0);

			\node[draw=none] at (-1.6,-1.3) {Optimal location:};
			
			\node [dot=$L$] at (0,-1.3) {};
			\node[draw=none] at (0,-1.6) {$\mathcal{M}(\mathbf{x},\mathbf{t})$};
			
			\node [dot=$L$] at (0.75,-1.3) {};
			\node[draw=none] at (1.75,-1.6) {4};
			
			\node [dot=$H$] at (3,-1.3) {};
			
			\node[draw=none] at (3.5,-1.6) {1.5};
			
			\node [dot=$H$] at (4,-1.3) {};
			\node[draw=none] at (4.5,-1.6) {1.5};
			
			\node [dot=$H$] at (5,-1.3) {};
			\draw (0,-1.3) -- (5,-1.3);		
			\end{tikzpicture}}
		\subfigure[Agent $a_4$ misreports her preference as ``L'']{
			\label{fig:exp2:b}
			\begin{tikzpicture}[scale=0.95, dot/.style={circle,inner sep=1pt,fill,label={#1},name={#1}},
			extended line/.style={shorten >=-#1,shorten <=-#1},
			extended line/.default=1cm]
			
			\node[draw=none] at (-1.7,-0) {Declared preferences:};
			
			\node [dot=$L$] at (0,0) {};
			\node[draw=none] at (0,-0.3) {$a_1$};
			\node[draw=none] at (0.35,-0.3) {1};
			
			\node [dot=$L$] at (0.75,0) {};
			\node[draw=none] at (0.75,-0.3) {$a_2$};
			\node[draw=none] at (1.75,-0.3) {4};
			
			\node [dot=$H$] at (3,0) {};
			\node[draw=none] at (3,-0.3) {$a_3$};
			
			\node[draw=none] at (3.5,-0.3) {1.5};
			
			\node [dot=$L$] at (4,0) {};
			\node[draw=none] at (4,-0.3) {$a_4$};
			\node[draw=none] at (4.5,-0.3) {1.5};
			
			\node [dot=$H$] at (5,0) {};
			\node[draw=none] at (5,-0.3) {$a_5$};
			\draw (0,0) -- (5,0);
			
			\node[draw=none] at (-1.7,-1.3) {Optimal location:};
			
			\node [dot=$H$] at (0,-1.3) {};
			\node[draw=none] at (0.35,-1.6) {1};
			
			\node [dot=$H$] at (0.75,-1.3) {};
			\node[draw=none] at (1.75,-1.6) {4};
			
			\node [dot=$L$] at (3,-1.3) {};
			
			\node[draw=none] at (3.5,-1.6) {1.5};
			
			\node [dot=$L$] at (4,-1.3) {};
			
			\node [dot=$L$] at (5,-1.3) {};
			\node[draw=none] at (5,-1.6) {$\mathcal{M}(\mathbf{x},\mathbf{t})$};
			\draw (0,-1.3) -- (5,-1.3);	
			\end{tikzpicture}}
		\caption{Instances in Theorem~\ref{thm:optimalLocation}}
		\label{fig:exp2}
	\end{figure}	
	We prove this theorem by the instances shown in Figure~\ref{fig:exp2} in which five agents are located in a line segment of a length of $8$km.  Figure~\ref{fig:exp2:a} shows the distance between agents and their true preferences. The optimal location of the facility is at $0$km (the location of agent $1$). No exchange would happen in the TTC algorithm as no  pair of agents would prefer exchanging the locations. The utility of agent $4$ is $6.5$. Now considering the case that agent $4$ falsely declares her preference as $L$ shown in Figure~\ref{fig:exp2:b}, the optimal location of the facility is at $8$km. Then in the TTC algorithm, agent $5$ would exchange her location with agent $1$, and agent $4$ would exchange her location with agent $2$. The utility of agent $4$ is $7$. This implies that agent $4$ benefit by manipulating her preference in the mechanism. Then the theorem directly follows.
\end{proof}

\subsection{Simple 2-approximation Mechanism}
Next we provide a simple universally truthful randomized mechanism that achieve at least half of the optimal social welfare in expectation. 
\begin{mechanism}{$\mathcal{M}$}\label{simpleMechanism}
	Given a location profile $\mathbf{x}$ and a preference profile $\mathbf{t}$, $\mathcal{M}^f(\mathbf{x},\mathbf{t}) = 0$ with probability $1/2$, and $\mathcal{M}^f(\mathbf{x},\mathbf{t}) = d$ with probability $1/2$. 
\end{mechanism}

\begin{theorem}
	Mechanism $\mathcal{M}$ is a universally truthful randomized mechanism that achieves $2$-approximation to the optimal social welfare.
\end{theorem}

\subsubsection{Universally truthfulness}
It is easy to verify that agents would not benefit from misreporting their preference as Mechanism $\mathcal{M}$ always places the facility at location $0$ with a probability of $1/2$ and places the facility at location $d$ with a probability of $1/2$. 

\subsubsection{Approximation}
The proof of the approximation of Mechanism $\mathcal{M}$ comes from two ingredients. 
\begin{lemma}
	\label{lem:nondecreasingSocialWelfare}
	During exchanging phases, the social welfare increases as agents exchange their locations with each other. 
\end{lemma} 
This lemma allows us to only consider the social welfare before agents exchanging their locations. In other word, if the social welfare before exchanging phases is close to the optimal social welfare, then the social welfare after exchanging phases is still close to the optimal social welfare. With some simple calculations, we get the following.

\begin{lemma}
	\label{lem:approximation}
	Given a location profile $\mathbf{x}$ and a preference profile $\mathbf{t}$, $\mathrm{SW}(\mathcal{M}^f(\mathbf{x},\mathbf{t}),  \mathbf{x},\mathbf{t}) = \frac{n\cdot d}{2}$.
\end{lemma}
\begin{proof}
	As Mechanism $\mathcal{M}$ is randomize mechanism, the expected social welfare it achieves can be written as follows.
	\begin{align*}
	& \mathrm{SW}(\mathcal{M}^f(\mathbf{x},\mathbf{t}),  \mathbf{x},\mathbf{t}) \\
	= & \mathrm{SW}(\mathbb{E}[\mathcal{M}^f(\mathbf{x},\mathbf{t})],\mathbf{x},\mathbf{t} )  \\
	= & \frac{1}{2} \mathrm{SW}(0,  \mathbf{x},\mathbf{t}) + \frac{1}{2} \mathrm{SW}(d,  \mathbf{x},\mathbf{t}) \\
	= & \frac{1}{2} \sum_{i \in A} u(0,x_i,t_i) + \frac{1}{2} \sum_{i \in A} u(d,x_i,t_i)
	\end{align*}
	By the fact that $u(0,x_i,t_i) + u(d,x_i,t_i) = d$ for any agent in $A$ regardless $x_i$ and $t_i$, it concludes that $\mathrm{SW}(\mathcal{M}^f(\mathbf{x},\mathbf{t}),  \mathbf{x},\mathbf{t}) = \frac{n\cdot d}{2}$. 
\end{proof}

By the fact that $\mathrm{SW}(\mathcal{M}(\mathbf{x}, \mathbf{t}), \mathcal{TTC}(\mathbf{x}, \mathcal{M}(\mathbf{x}, \mathbf{t}),\mathbf{t}), \mathbf{t}) > \mathrm{SW}(\mathcal{M}(\mathbf{x},\mathbf{t}),\mathbf{x},\mathbf{t})$ and the utility of each agent is at most $d$, the approximation of Mechanism $\mathcal{M}$ directly follows.

\section{Conclusion}
In this paper, we consider the mechanism design without payments when selfish agents have difference types of preferences. We identify a limitation of the classical approach, that is, agents could benefit by exchanging their allocations obtained from mechanisms. To cope with this issue, we present a model called mechanism design with exchangeable allocations. We consider two approaches for this model, central exchanges and individual exchanges. By using facility location games as an example, we provide a truthful mechanism that achieves the optimal social welfare in central exchanges. We also provide a universally truthful randomized mechanism that achieves at least half of the optimal social welfare in individual exchanges. It will be interesting to investigate the problem in  settings other than facility location games. We leave it as future work. 

\newpage
\bibliographystyle{named}
\bibliography{ref}

\end{document}